\documentclass[letterpaper, 10pt, conference]{ieeeconf}
\IEEEoverridecommandlockouts 

\usepackage{cite}
\usepackage{algorithmic}
\usepackage{textcomp}
\usepackage{latexsym,amsmath,amssymb,amsfonts}

\usepackage{amsthm}
\usepackage[utf8]{inputenc}
\usepackage{rotating}
\usepackage{epstopdf}
\usepackage[dvipsnames]{xcolor}
\usepackage[normalem]{ulem}
\usepackage{caption}
\usepackage{subcaption}
\usepackage{fancybox}
\usepackage{colortbl}
\usepackage{float}
\usepackage{comment}
\usepackage{mathdots}
\usepackage{dblfloatfix}
\usepackage{graphicx}
\usepackage{placeins}
\usepackage{marginnote}
\usepackage{bm}
\usepackage{url}
\usepackage{empheq}
\let\labelindent\relax
\usepackage{enumitem}
\usepackage[hidelinks]{hyperref}
\usepackage{tikz}
\usetikzlibrary{arrows,shapes,positioning}
\usetikzlibrary{decorations.markings,decorations.pathmorphing,
decorations.pathreplacing}
\usetikzlibrary{calc,patterns,shapes.geometric}
\usetikzlibrary{shapes.multipart,positioning,patterns,backgrounds}

\newtheorem{theorem}{\it \textbf{Theorem}}
\newtheorem{lemma}{\it \textbf{Lemma}}
\newtheorem{definition}{\it \textbf{Definition}}
\newtheorem{remark}{\it \textbf{Remark}}

\newtheorem{proposition}{\it \textbf{Proposition}}
\newtheorem{corollary}{\it \textbf{Corollary}}
\newtheorem{fact}{\it \textbf{Fact}}
\definecolor{cmass}{RGB}{31,111,63}   
\definecolor{ckill}{RGB}{194,26,139}  
\definecolor{cact}{RGB}{138,75,0}
\definecolor{cgrey}{RGB}{90,90,90}
\definecolor{cblue}{RGB}{40,80,160}

\newcommand{\E}{\mathbb{E}}
\newcommand{\R}{\mathbb{R}}

\newcommand{\calP}{\mathcal{P}}
\newcommand{\calU}{\mathcal{U}}

\newcommand{\Lbar}{\Lambda}
\newcommand{\calR}{\mathcal{R}}
\definecolor{revcol}{RGB}{160,40,40}    
\newcommand{\rev}[1]{#1}

\definecolor{edcol}{RGB}{0,90,180}
\definecolor{plcol}{RGB}{0,120,80}
\newcommand{\pl}[1]{#1}
\newcommand{\ed}[1]{#1}
\newcommand{\edm}[1]{#1}
\title{\LARGE \bf
Removal-Only Actuation in Age-Structured Branching Populations: Fundamental Limits of Equilibrium Placement
}
\author{Ouerdia Arezki$^{1}$, Ali Zemouche$^{2}$
\thanks{The authors gratefully acknowledge the support of the IUT Henri Poincar\'e de Longwy, Universit\'e de Lorraine. This work was initiated while O. Arezki was a visiting professor hosted by CRAN (CNRS UMR 7039) and the IUT Henri Poincar\'e de Longwy.}
\thanks{A short version of this paper, restricted to
Sections~\ref{sec:main}--\ref{sec:prop}, has been submitted to the IEEE Control
Systems Letters (L-CSS).}
\thanks{$^{1}$Laboratoire de Math\'ematiques Blaise
Pascal (LMBP), UMR CNRS 6620, Universit\'e Clermont Auvergne, F-63178
Aubi\`ere, France (email: \protect\url{ouerdia.arezki@uca.fr})}
\thanks{$^{2}$Universit\'e de Lorraine, CRAN CNRS UMR 7039, 54400 Cosnes
et Romain, France~(e-mail: \protect\url{ali.zemouche@univ-lorraine.fr}).}
}

\makeatletter
\renewenvironment{proof}[1][\proofname]{\par
  \pushQED{\qed}%
  \normalfont\topsep6\p@\@plus6\p@\relax
  \trivlist
  \item[\hskip\labelsep\itshape #1\@addpunct{.}]\ignorespaces
}{%
  \popQED\endtrivlist\@endpefalse
}

\makeatother

\title{
\rev{Removal-Only Actuation in Age-Structured Branching Populations: Fundamental Limits of Equilibrium Placement}
}
\author{Ouerdia Arezki$^{1}$, Ali Zemouche$^{2}$
\thanks{The authors gratefully acknowledge the support of the IUT Henri Poincar\'e de Longwy, Universit\'e de Lorraine. This work was initiated while O. Arezki was a visiting professor hosted by CRAN (CNRS UMR 7039) and the IUT Henri Poincar\'e de Longwy.}
\thanks{\ed{A short version of this paper, restricted to
Sections~\ref{sec:main}--\ref{sec:prop}, has been submitted to the IEEE Control
Systems Letters (L-CSS).}}
\thanks{$^{1}$Laboratoire de Math\'ematiques Blaise
Pascal (LMBP), UMR CNRS 6620, Universit\'e Clermont Auvergne, F-63178
Aubi\`ere, France (email: \protect\url{ouerdia.arezki@uca.fr})}
\thanks{$^{2}$Universit\'e de Lorraine, CRAN CNRS UMR 7039, 54400 Cosnes
et Romain, France~(e-mail: \protect\url{ali.zemouche@univ-lorraine.fr}).}
}

\makeatletter
\renewenvironment{proof}[1][\proofname]{\par
  \pushQED{\qed}%
  \normalfont\topsep6\p@\@plus6\p@\relax
  \trivlist
  \item[\hskip\labelsep\itshape #1\@addpunct{.}]\ignorespaces
}{%
  \popQED\endtrivlist\@endpefalse
}

\makeatother

\begin{document}
\maketitle
\begin{abstract}
A subcritical age-structured branching population dies out almost surely.
Conditioned on survival, it converges to its Yaglom limit, a quasi-stationary
equilibrium that we take as the operating point for control. We model preventive
removal (culling) as an age-dependent actuator \rev{that raises the mortality
rate and leaves the offspring law untouched}, and we show that its authority
over this equilibrium is \pl{bounded for structural reasons}. \rev{Two facts
drive the result. First, the input is matched to the killing rate but unmatched with
respect to the Foster--Lyapunov drift, so the transmission barrier $\Lbar$ set
by reproduction alone is invariant under such actuation. Second\ed{, and this does
not follow from invariance alone,} the supremum of the reachable decay rates is
$\Lbar+\nu^\star$, where $\nu^\star\le0$ is the Malthusian parameter of the
lineage conditioned never to die childless; the gap $|\nu^\star|$ is given in
closed form and vanishes exactly when no individual has two or more offspring.
Consequently no removal law of this class reaches the barrier, and along the
admissibility boundary the achievable decay rate is governed by the shape of the
actuator rather than by its \pl{size}.} We illustrate these results on a
model calibrated to the 2001 Cumbrian foot-and-mouth outbreak.
\end{abstract}


\ed{\textbf{Keywords:}  Admissible controls, age-structured populations, equilibrium placement, fundamental limitations, monotone systems, reachable sets, stochastic systems.}
\section{Introduction}
An age-structured branching population almost surely dies out in the subcritical regime. Its long-term object is therefore not a stationary distribution but the quasi-stationary distribution (QSD), or Yaglom limit, namely the law of
the population conditioned on survival~\cite{bellman,athreya,tran,yaglom,meleard}. In applications such as epidemic surveillance and invasive-species management, this conditioned law provides the natural equilibrium for control.
In~\cite{companion}, we characterize this equilibrium by proving the existence and uniqueness of the QSD together with the convergence of a Fleming--Viot particle approximation for the uncontrolled plant. \rev{The present paper introduces
preventive removal as a control input and asks how far this actuator can displace
the conditioned equilibrium. The answer is a structural limitation of a specific
actuator class, in the sense of feedback theory~\cite{sbg,freudenberg,stein},
rather than a new probabilistic estimate.}

\rev{The control problem is unusual: \pl{the plant needs no stabilisation}.} It contracts to its conditioned
equilibrium on its own, so the command places the attractor rather than stabilising it. The state is a measure-valued population, observed only through a Fleming--Viot particle estimator. The removal hazard $u$ is the
control input, the regulated output is the decay rate $\lambda_0=\nu_Y(\kappa)$ of the conditioned equilibrium, and a transmission barrier $\Lbar$ bounds the reachable decay rate.

\rev{Everything rests on one modelling hypothesis, which we state at the outset
because it \pl{limits} the scope of every result below: \emph{the actuator is
childless}, i.e.\ it raises the mortality rate and does not alter the offspring
law $(p_n)$. Removal terminates a lineage; it does not make transmission less
likely. Under this hypothesis the barrier $\Lbar$, which is determined by the
uncontrolled lifetime and offspring laws, is invariant (Lemma~\ref{lem:invariance}), so
$\lambda_0(u)<\Lbar$ is immediate. What is \pl{not} immediate\ed{, and is the
mathematical content of Theorem~\ref{thm:sup},} is the exact value of the
supremum, $\Lbar+\nu^\star$, which is strictly below $\Lbar$ by a computable
margin. \pl{Actuators that also reduce transmission act on $\Lbar$ itself and
are not covered here.}}

\pl{This work draws on three lines of research.} Fundamental limitations of
feedback, in the Bode--Freudenberg--Looze
tradition~\cite{bode,freudenberg,sbg,stein}, \pl{give the setting}: a bound on
achievable performance that no control law \rev{in a given class} can \pl{avoid}, because it is a
property of the plant--actuator pair. Positive and monotone
systems~\cite{farina,angeli,rantzer} \pl{provide the ordering tools} behind the
monotone input--output map established below. Epidemic control over
networks~\cite{preciado,nowzari,pare} and the optimal control of age-dependent
population PDEs~\cite{anita,webb,lenhart} \pl{provide the application context.
What is different here is} that the operating point is a \pl{conditioned}
law, not a stationary one, and is available only through particle estimation.

\pl{The contributions of this paper are the following.} We prove that the input is matched to the killing rate
but unmatched with respect to the Foster--Lyapunov drift
(Lemma~\ref{lem:invariance}); every limitation below follows from this single
mismatch. We derive a closed-form, state-independent input constraint set
(Theorem~\ref{thm:envelope}), necessary and independent of the control law, and
sufficient for the proportional family. We show that the input--output map is
monotone and characterise the reachable output set together with its supremum
(Theorem~\ref{thm:sup}), \rev{which is strictly sharper than the elementary
bound $\lambda_0(u)<\Lbar$ inherited from invariance: it identifies the
supremum exactly and quantifies the residual gap in closed form}. We show that\rev{, along the
admissibility boundary,} achievable performance is set by
the input \pl{shape}, not its magnitude (Corollary~\ref{cor:shape}). For
proportional actuation we give an exact certifiable gain interval
(Theorem~\ref{thm:cmax}), \pl{show how fast the certified stability margin
shrinks and what this costs in sensing} near the constraint boundary
(Proposition~\ref{prop:cost}), and prove that the certainty-equivalence gain law
driven by the particle estimator is ISS with respect to the estimation error
(Proposition~\ref{prop:feedback}).



\section{Preliminaries and problem formulation}
\label{sec:plant}
This section recalls the notation and results of~\cite{companion} used
throughout the paper and states the control problem. From a control
perspective, the plant is an autonomous, \ed{open-loop convergent} conditioned
process whose operating point is the Yaglom limit.

\subsection{\rev{Measure-valued plant model}}\label{sec:measure}
Let $G$ be a lifetime distribution on $[0,\infty)$ with continuous density
$g$ and $G(0)=0$, and let
$h(s)=\sum_{n\ge0}p_ns^n$ denote the offspring generating function, where
$\xi\sim(p_n)$, $m:=h'(1)<\infty$, and
$p_0:=\mathbb P(\xi=0)$. Each individual lives for a random lifetime
$T\sim G$ and, upon death, is independently replaced by $\xi$ offspring born
at age zero. Write
$\bar G:=1-G$, $\mu:=g/\bar G$ for the hazard rate,
$\mu^\ast:=\sup_a\mu(a)$, and
$\mu_\infty:=\liminf_{a\to\infty}\mu(a)\le\mu^\ast$.
Because the death intensity
$\sum_i\mu(a_i(t))$ depends on the age configuration rather than on the
population size, the size process is not Markovian. The Markov property is
recovered on the space $\mathcal M_p(\R_+)$ of finite point measures,
equipped with the narrow topology~\cite{tran}. Defining
\begin{equation*}
\eta_t:=\sum_i\delta_{a_i(t)},\;
Z_t:=\langle\eta_t,\mathbf1\rangle,\;
\tau:=\inf\{t>0:\eta_t=0\},
\end{equation*}
and setting $E:=\mathcal M_p(\R_+)\setminus\{0\}$, the extended
generator~\cite{meyntweedie} is, for
$F(\eta)=\Phi(\langle\eta,f\rangle)$,
\begin{multline}
\mathcal{L}F(\eta)=
\Phi'(\langle\eta,f\rangle)\langle\eta,f'\rangle
+\int_{\R_+}\mu(a)\sum_{n\ge0}p_n\\
\Bigl[
\Phi\bigl(\langle\eta,f\rangle-f(a)+nf(0)\bigr)
-\Phi\bigl(\langle\eta,f\rangle\bigr)
\Bigr]\eta(da),
\label{eq:gen}
\end{multline}
where individuals age at unit speed between jumps, and an individual of age
$a$ is replaced at rate $\mu(a)$ according to
$\eta\mapsto\eta-\delta_a+\xi\delta_0$. For $\Phi=\mathrm{id}$,
\[
\mathcal L\langle\eta,f\rangle=\langle\eta,Af\rangle,
\;
(Af)(a)=f'(a)+\mu(a)\bigl(mf(0)-f(a)\bigr).
\]
Extinction occurs only when a singleton dies without offspring, so the killing
rate is
\begin{equation}
\kappa(\eta)=p_0\,\mu(a)\,\mathbf1_{\{\eta=\delta_a\}},
\;
\|\kappa\|_\infty=p_0\mu^\ast<\infty.
\label{eq:kappa}
\end{equation}
\pl{Absorption is a boundary that the size process cannot cross; on the lifted
state space it becomes a bounded killing mechanism.} Let $X$ denote the process on $E$ obtained
by removing the transition $\delta_a\mapsto0$. Then $\eta$ is $X$ killed at rate
$\kappa$, with generator
$\mathcal L_XF=\mathcal LF+\kappa(F-F(0))$. Under $X$, a singleton can leave the
large-age region only through successful reproduction, at rate $(1-p_0)\mu(a)$,
so the maximal drift rate available in a Foster--Lyapunov argument is
\begin{equation}
\Lbar:=(1-p_0)\mu_\infty,
\label{eq:barrier}
\end{equation}
rather than $\mu_\infty$. We call $\Lbar$ the transmission barrier; it is
determined solely by the reproduction mechanism. Finally, the Malthusian
parameter $\lambda_0$ is the unique solution of the Euler--Lotka equation
$m\int_0^\infty e^{\lambda t}\,dG(t)=1$.
\rev{We use the decay-rate sign convention throughout: $\lambda_0>0$ in the
subcritical regime, and $-\lambda_0$ is the exponential growth rate.}

\subsection{\rev{Standing assumptions and prior results}}\label{sec:background}
\rev{We collect here the conditions used in the sequel. Each statement
below specifies the subset it requires.}
\begin{itemize}[leftmargin=2.2em,itemsep=1pt]
\item[(H1)] $G$ is absolutely continuous with continuous density $g$,
$G(0)=0$, $\operatorname{supp}(G)=[0,\infty)$, and
$\inf_{[\ell,L]}g>0$ for every $0<\ell<L<\infty$;
\item[(H2)] $m<1$;
\item[(H3)] $p_1>0$;
\item[(H4)] $0<\mu_\infty\le\mu^\ast<\infty$;
\item[(H5)] $p_0\mu^\ast<\Lbar$;
\item[(\rev{H6})] $\sum_n n^\alpha p_n<\infty$ for some $\alpha>1$ satisfying
$\alpha\lambda_0>\Lbar$.
\end{itemize}
Conditions (H1)--(H5) are inherited from~\cite{companion}\rev{: (H2) fixes the
subcritical regime, (H5) is the \pl{main} structural requirement, and (H1),
(H3) and~(H4) are regularity and non-degeneracy conditions on the lifetime and
offspring laws. Assumption~(H6) is what the controlled problem adds:} it
strengthens the companion condition
$\alpha\lambda_0>p_0\mu^\ast=\|\kappa\|_\infty$ to $\alpha\lambda_0>\Lbar$,
so that a single exponent $\alpha$ serves for every admissible input, the
killing rate being allowed to rise up to $\Lbar$. It is not restrictive in the
calibrated application, whose offspring distribution has moments of every order.
Moreover, (H5) implies $p_0<1/2$, hence $m>1/2$.

\ed{The certificates used throughout the paper are the following.}
\begin{definition}[\ed{Admissible certificate}]
\label{def:cert}
\ed{Let $u\in L^\infty_+(\R_+)$. An \emph{admissible certificate} for $X_u$ is
a triple $(V,\lambda_1,C)$ in which $V:E\to[0,\infty)$ belongs to the extended
domain of $\mathcal L_{X_u}$, has relatively compact sublevel sets, satisfies
$B_V:=\sum_{n\ge1}p_nV(n\delta_0)<\infty$ and has $W:=V(\delta_\cdot)$ of class
$C^1$ and eventually non-decreasing, and in which $C<\infty$ and
$\lambda_1>\|\kappa_u\|_\infty$ satisfy}
\[
\ed{\mathcal L_{X_u}V\le-\lambda_1V+C \qquad\text{on }E .}
\]
\ed{We call $\lambda_1$ the certified drift rate.}
\end{definition}
\rev{Three results of~\cite{companion} are used repeatedly below, and we state
them without proof.}
\begin{fact}[Malthusian parameter and reproductive value {\rm\cite[Lem.~2.2]{companion}}]
\label{fact:malthus}
Under {\rm(H1), (H2), (H4)} and {\rm(H5)}:
\begin{enumerate}\itemsep1pt
\item[(i)] $\int e^{\theta a}\,dG(a)<\infty$ for every $\theta<\mu_\infty$;
\item[(ii)] the Euler--Lotka equation has a unique solution
$0<\lambda_0<\Lbar$;
\item[(iii)] the reproductive value $v(a)=m\,\E[e^{\lambda_0(T-a)}\mid T>a]$
belongs to $C_b^1(\R_+)$, satisfies $v(0)=1$, $m\le v\le v_{\max}$, and
$Av=-\lambda_0v$.
\end{enumerate}
The stable-age profile and reproductive value coincide with the eigenelements of
the associated renewal semigroup~\cite[Thm.~3.8]{bansaye2020}.
\end{fact}
\begin{fact}[Necessity of the transmission barrier {\rm\cite[Prop.~3.2]{companion}}]
\label{fact:necessity}
Assume {\rm(H2)} and {\rm(H4)}. \ed{If the uncontrolled process $X$ admits an
admissible certificate $(V,\lambda_1,C)$ in the sense of
Definition~\ref{def:cert} with $u\equiv0$, then necessarily
$\lambda_1\le\Lbar$, so that {\rm(H5)} is necessary. It is also sufficient:
under {\rm(H1)--(H6)} the certificate underlying Fact~\ref{fact:selection}
is admissible, so the class is non-empty.}
\end{fact}
\begin{fact}[Selection, contraction, and estimability {\rm\cite[Thms.~3.6, 3.8, Cor.~3.9]{companion}}]
\label{fact:selection}
Under {\rm(H1)--(H6)}, Assumption~H of~\cite{cv1} holds for $X$. The killed
semigroup admits a unique QSD $\nu_{\mathrm{QSD}}=\nu_Y$, with decay rate
$\lambda_0$, and
\[
\|\mathbb P_\varrho(\eta_t\in\cdot\mid\tau>t)-\nu_Y\|_{\mathrm{TV}}
\le C_0\varrho(V)e^{-\gamma t}.
\]
Moreover, the stationary empirical measure of the $N$-particle Fleming--Viot
system~\cite{villemonais} satisfies
\begin{equation}
\E\!\left|\mathcal X^N(\varphi)-\nu_Y(\varphi)\right|
\le dN^{-\varpi}\|\varphi\|_\infty,
\;
\varpi=\frac{\gamma}{2(\|\kappa\|_\infty+\gamma)},
\label{eq:rate}
\end{equation}
where $d$ depends on the uniform moment estimate
$\sup_N\E[\mathcal X^N(V)]\le C/(\lambda_1-\|\kappa\|_\infty)$. Finally,
\begin{equation}
\lambda_0=\nu_Y(\kappa),
\label{eq:selfcons}
\end{equation}
a self-consistency identity showing that $\lambda_0$ is an estimable output.
\end{fact}
Two consequences of~\eqref{eq:selfcons} are used repeatedly:
\begin{equation}
\lambda_0=\nu_Y(\kappa)\le\|\kappa\|_\infty,
\label{eq:lam0kappa}
\end{equation}
with strict inequality whenever $\mathbb P(\xi\ge2)>0$, since $\nu_Y$ then
charges $\{Z\ge2\}$; and every quasi-stationary distribution $\nu$ satisfies
$\nu(\kappa)\le\|\kappa\|_\infty$.

\subsection{\rev{Control problem}}\label{sec:problem}
Preventive removal (culling) acts as an additional age-dependent hazard
$u(a)\ge0$: an individual removed this way terminates its lineage, whereas a
natural death replaces it by $\xi\sim(p_n)$. Write $X_u$ for the corresponding
unkilled process, obtained by deleting the transition $\delta_a\mapsto0$, so
that $u$ enters only through the singleton killing rate $\kappa_u=p_0\mu+u$. The
plant is the controlled conditioned flow
$\Sigma:\ \dot\varrho_t=\mathcal L_u^{\ast}\varrho_t$ on $\calP(E)$, with input
$u\in L^\infty_+(\R_+)$. It is autonomous and \ed{open-loop convergent}: whenever an
admissible Foster--Lyapunov certificate exists (Fact~\ref{fact:necessity}), the
flow contracts to a unique attractor $\nu_Y(u)$ at rate $\gamma(u)$. The
regulated output is the decay rate of that attractor,
$y(u)=\lambda_0(u)=\nu_Y(u)(\kappa_u)$, which by~\eqref{eq:selfcons}
and~\eqref{eq:rate} is read from the particle system through
$\widehat y_N=\mathcal X^N(\kappa_u)$, with
$\E|\widehat y_N-y|\le d\,\|\kappa_u\|_\infty N^{-\varpi}$. The input is
constrained to the \rev{\emph{certifiable} class}
\[
\rev{\calU_{\mathrm c}}:=\{u\in L^\infty_+(\R_+):\text{$X_u$ admits an admissible certificate}\}
\]
\ed{in the sense of Definition~\ref{def:cert},}
\rev{for which Theorem~\ref{thm:envelope} gives an explicit outer bound.}
The control problem is equilibrium placement under certifiability~(Fig.~\ref{fig:loop}):
\begin{enumerate}[leftmargin=2.2em,itemsep=1pt]
\item[(i)] \ed{bound} the reachable output set
$\calR:=\{\lambda_0(u):u\in\calU\rev{_{\mathrm c}}\}$ \ed{and locate its
supremum};
\item[(ii)] for a target $\lambda^\star$ in the interior of $\calR$, construct
an output-feedback law $u=\mathcal K(\widehat y_N)$ placing the attractor at
$\nu_Y(\lambda^\star)$, and bound the placement error in $N$.
\end{enumerate}
No stabilisation is required, the flow converges on its own, so the design question is which equilibrium can be assigned, and at what certified cost. The
reachable-set answer is a limitation of the plant--actuator pair, in the spirit
of the fundamental limitations of feedback~\cite{bode,freudenberg,sbg}.
\begin{figure}[h!]
\centering
\begin{tikzpicture}[>=stealth',
  blk/.style={draw,rounded corners=1pt,align=center,font=\scriptsize,
              minimum height=7.5mm,inner sep=2.5pt}]
\node (sum) [draw,circle,inner sep=0.8pt,font=\scriptsize] at (0,0) {$-$};
\node (ref) [font=\tiny,left=3.5mm of sum] {$\lambda^\star$};
\node (inv) [blk,right=3.5mm of sum,fill=ckill!8,text width=17mm]
      {feedback law\\ $u=\mathcal K(\widehat y_N)$};
\node (act) [blk,right=3mm of inv,text width=12mm] {actuator\\ $u\in\rev{\calU_{\mathrm c}}$};
\node (pl)  [blk,right=3mm of act,fill=cmass!10,text width=21mm]
      {conditioned flow\\ contracts at $\gamma(u)$\\ to $\nu_Y(u)$};
\node (est) [blk,below=6mm of pl,text width=32mm]
      {Fleming--Viot estimator, $N\ge2$\\
       $\widehat y_N=\mathcal{X}^N(\kappa_u)$, error $O(dN^{-\varpi})$};
\draw[->] (ref) -- (sum);   \draw[->] (sum) -- (inv);
\draw[->] (inv) -- (act);   \draw[->] (act) -- (pl);
\draw[->] (pl) -- (est);
\draw[->] (est) -| (sum);
\end{tikzpicture}
\caption{\rev{Control structure. The removal input $u$ enters the conditioned flow
through the killing rate $\kappa_u$; the regulated output $\lambda_0(u)$ is
estimated by the Fleming--Viot particle system ($N\ge2$) and compared with the
target $\lambda^\star$.}}
\label{fig:loop}
\end{figure}
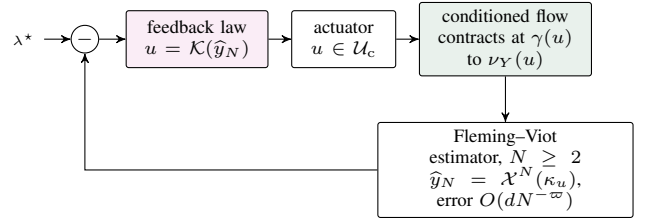



\section{Main results: Reachable decay rates and maximal actuator authority}
\label{sec:main}
\rev{This section develops the control theory of childless actuation. The
following lemma isolates the single structural mismatch from which every
limitation below is derived: the input enters the killing channel only, and the
generator of the unkilled dynamics, which carries the Foster--Lyapunov
drift, is left invariant.}
\begin{lemma}
\label{lem:invariance}
\rev{Let $u\ge0$ be bounded and measurable. Then:}
\begin{enumerate}\itemsep1pt
\item[(i)]
the mortality rate of the controlled process is
\begin{align*}
\kappa_u(\eta)
&=
\bigl[p_0\mu(a)+u(a)\bigr]\mathbf1_{\{\eta=\delta_a\}},\\
\|\kappa_u\|_\infty
&=
\sup_{a\ge0}\bigl[p_0\mu(a)+u(a)\bigr];
\end{align*}
\item[(ii)]
the generator of the unkilled dynamics is invariant on the singleton stratum:
\[
\mathcal L_{X_u}F(\delta_a)
=
\mathcal L_XF(\delta_a),
\qquad a\ge0,
\]
for every $F$ in the extended domain.
\end{enumerate}
\end{lemma}
\begin{proof}
For a singleton $\eta=\delta_a$, the control input enters the extended
generator~\eqref{eq:gen} only through the additional removal term
$u(a)[\Phi(0)-\Phi(f(a))]$, yielding
\begin{equation}\label{eq:Lu-singleton}
\begin{split}
\mathcal L^uF(\delta_a)
={}&\Phi'(f(a))\,f'(a)\\
&+\mu(a)\sum_{n\ge0}p_n\bigl[\Phi(nf(0))-\Phi(f(a))\bigr]\\
&+u(a)\bigl[\Phi(0)-\Phi(f(a))\bigr].
\end{split}
\end{equation}
Absorption occurs only from singleton states, through either the childless event,
at rate $p_0\mu(a)$, or controlled removal, at rate $u(a)$. Hence
\[
\kappa_u(\delta_a)=p_0\mu(a)+u(a),
\]
whereas $\kappa_u$ vanishes outside the singleton stratum. This proves~(i).
Removing the absorbing transition from~\eqref{eq:Lu-singleton} yields
\[
\mathcal L_{X_u}F(\delta_a)
=
\Phi'(f(a))\,f'(a)
+\mu(a)\sum_{n\ge1}p_n
\bigl[\Phi(nf(0))-\Phi(f(a))\bigr],
\]
which is independent of $u$. Therefore,
\[
\mathcal L_{X_u}F(\delta_a)
=
\mathcal L_XF(\delta_a),
\]
which proves~(ii).
\end{proof}
Lemma~\ref{lem:invariance} shows that childless actuation can only increase the
killing rate $\kappa_u$, thereby tightening the drift condition
$\lambda_1>\|\kappa_u\|_\infty$ without modifying the drift itself. The
following theorem translates this structural property into a closed-form,
control-law-independent characterization of the admissible actuator class.
\begin{theorem}
\label{thm:envelope}
\ed{Assume {\rm(H2)} and {\rm(H4)}}, and let $u\ge0$ be bounded \rev{and}
measurable. \ed{If $X_u$ admits an admissible certificate $(V,\lambda_1,C)$
in the sense of Definition~\ref{def:cert}, then}
\begin{equation}
\sup_{a\ge0}\bigl[p_0\mu(a)+u(a)\bigr] < \lambda_1
\le
\Lbar,
\label{eq:env0}
\end{equation}
\rev{in particular,}
\begin{equation}
u(a)  < (1-p_0)\mu_\infty-p_0\mu(a),
\qquad a\ge0.
\label{eq:envelope}
\end{equation}
\rev{Since $u$ is arbitrary, $\calU_{\mathrm c}\subseteq\calU$, where}
\begin{equation}
\rev{\calU:=\{u\ge0:\|\kappa_u\|_\infty<\Lbar\}.}
\label{eq:envclass}
\end{equation}
\end{theorem}
\begin{proof}
By Lemma~\ref{lem:invariance}(ii), the singleton drift is independent of the
control:
\[
\mathcal L_{X_u}V(\delta_a)
=
W'(a)
-(1-p_0)\mu(a)W(a)
+\mu(a)B_V.
\]
As in Fact~\ref{fact:necessity},
\[
[\lambda_1-(1-p_0)\mu(a)]W(a)\le C
\]
for sufficiently large $a$, whereas $W(a)\to\infty$ because the sublevel sets of
$V$ are relatively compact. If $\lambda_1>\Lbar$, choose
$a_k\to\infty$ such that $\mu(a_k)\to\mu_\infty$. Then
\[
\lambda_1-(1-p_0)\mu(a_k)\longrightarrow\lambda_1-\Lbar>0,
\]
so $W(a_k)$ remains bounded, a contradiction. Hence
$\lambda_1\le\Lbar$. Together with
$\lambda_1>\|\kappa_u\|_\infty$, this yields~\eqref{eq:env0}, from which
\eqref{eq:envelope} follows.
\end{proof}

\rev{Theorem~\ref{thm:envelope} gives necessity only. The converse inclusion $\calU\subseteq\calU_{\mathrm c}$ holds on the hazard-proportional family (Theorem~\ref{thm:cmax}) and is open in general, so every statement below bounds certifiable performance from above.}

\rev{\pl{One could try to} enlarge $\calU$ by weakening the drift requirement $\lambda_1>\|\kappa_u\|_\infty$ to $\lambda_1>\operatorname{osc}\kappa_u:=\sup\kappa_u-\inf\kappa_u$, as allowed by~\cite[Rem.~1]{cv1}. This brings no improvement here: absorption occurs only from singletons, so $\kappa_u$ vanishes on $\{Z\ge2\}$, \pl{and therefore} $\inf\kappa_u=0$ and $\operatorname{osc}\kappa_u=\|\kappa_u\|_\infty$ for every $u$, including $u\equiv0$. \pl{The envelope~\eqref{eq:envelope} is therefore not caused by the particular norm used in the certificate.}}

\rev{The envelope tightens as the hazard increases, so the available authority
decreases with age and is smallest at ages where the natural childless rate
$p_0\mu(a)$ leaves the least room below $\Lbar$.}

Moreover, the controlled process generally no longer belongs to the Bellman--Harris class, since the effective offspring mean $m\mu(a)/[\mu(a)+u(a)]$ becomes age-dependent. Only proportional actuation $u=c\mu$ preserves the Bellman--Harris structure. Establishing sufficiency for general control requires the controlled reproductive value satisfying
\begin{equation}
A_uf=f'+m\mu f(0)-(\mu+u)f=-\lambda_0(u)f,
\label{eq:agecancel}
\end{equation}
which we derive only for the proportional family in
Section~\ref{sec:prop}.
Lemma~\ref{lem:invariance} also shows that the transmission barrier is invariant
under control and therefore cannot be reached. It remains to determine how
closely the controlled decay rate $\lambda_0(u)$ can approach this barrier.
The controlled decay rate is characterized through the first-moment dynamics.
An individual born at age $0$ produces offspring at age $a$ at rate
$m\mu(a)\bar G_u(a)$, where
\[
\bar G_u(a)=\exp\!\left(-\int_0^a(\mu+u)\right),
\]
so that the first-moment decay rate is the unique root of
\begin{equation}
\Psi_u(\lambda)
:=
m\int_0^\infty
e^{\lambda a}\mu(a)\bar G(a)e^{-\int_0^au}\,da
=1.
\label{eq:ELu}
\end{equation}
For $u=0$, this reduces to the Euler--Lotka equation, whereas for
$u=c\mu$ it becomes
\[
m\int e^{\lambda a}g\,\bar G^{\,c}da=1.
\]
\rev{We write $\lambda_0(u)$ for this root. It is the decay rate
$\nu_Y(u)(\kappa_u)$ of the conditioned equilibrium whenever
Fact~\ref{fact:selection} applies, in particular on the proportional family of
Section~\ref{sec:prop} by Theorem~\ref{thm:cmax}; otherwise
Theorem~\ref{thm:sup} is a statement about the first moment.}
\begin{theorem}
\label{thm:sup}
Assume {\rm(H1), (H2), (H4)} \rev{and {\rm(H5)}. For every measurable $u\ge0$
with $\|\kappa_u\|_\infty\le\Lbar$\ed{, in particular for every $u\in\calU$,}
equation~\eqref{eq:ELu} has a unique root $\lambda_0(u)$, and
$0<\lambda_0(u)\le\Lbar$. Moreover:}
\begin{enumerate}\itemsep1pt
\item[(i)]
\rev{If $u_1,u_2\ge0$ satisfy
$\|\kappa_{u_i}\|_\infty\le\Lbar$ $(i=1,2)$
and
$u_1\le u_2$ pointwise}, then
$\lambda_0(u_1)\le\lambda_0(u_2)$, with strict inequality whenever
$u_1<u_2$ on a set of positive Lebesgue measure.
\item[(ii)]
\rev{Let $u^\star(a):=\Lbar-p_0\mu(a)$. Every $u\in\calU$ satisfies
$u<u^\star$ pointwise, and $u^\star\notin\calU$, yet}
\begin{equation}
\sup_{u\in\calU}\lambda_0(u)=\lambda_0(u^\star),
\label{eq:sup}
\end{equation}
\rev{a supremum that is not attained.}
\item[(iii)]
\rev{Let $\nu^\star$ be the unique root of}
\begin{equation}
m\int_0^\infty
e^{\nu a}\mu(a)\bar G(a)^{1-p_0}\,da=1.
\label{eq:nustar}
\end{equation}
\rev{Then $\lambda_0(u^\star)=\Lbar+\nu^\star$, and $\nu^\star\le0$ with
$\nu^\star=0$ if and only if $\mathbb P(\xi\ge2)=0$.}
\end{enumerate}
\rev{Combining~(ii), (iii) and Theorem~\ref{thm:envelope},}
\[
\lambda_0(u)\ <\ \Lbar+\nu^\star\ \le\ \Lbar
\qquad\text{for every }u\in\calU_{\mathrm c}:
\]
\rev{certifiable childless actuation never reaches the transmission barrier,
and falls short of it by more than the deficit $|\nu^\star|$.}
\end{theorem}
\begin{proof}
\rev{Let $u\ge0$ be measurable with $u\le u^\star$. Since
$m\ge1-p_0$, assumption~(H2) implies $p_0>0$, and therefore
$\Lbar=(1-p_0)\mu_\infty<\mu_\infty$ by~(H4). Moreover,
\[
\Psi_u(\lambda)
\le
\Psi_0(\lambda),
\]
so Fact~\ref{fact:malthus}(i) ensures that $\Psi_u$ is finite
and continuous on $[0,\Lbar]$. 
It is also strictly increasing:
\ed{for $\lambda<\lambda'$ one has $e^{\lambda a}<e^{\lambda'a}$ at every $a>0$,
and by~(H1) the measure $m\mu\bar Ge^{-\int_0^\cdot u}\,da$ charges
$(0,\infty)$.}
Furthermore,
\begin{align*}
\Psi_u(0)
&=
m\int_0^\infty
\mu(a)\bar G(a)e^{-\int_0^au}\,da \\
&\quad\le
m\int_0^\infty
\underbrace{\mu(a)\bar G(a)}_{g(a)}\,da
=
m<1.
\end{align*}
On the other hand, integrating $u\le u^\star:= \Lbar-p_0\mu(a)$ yields
\[
e^{-\int_0^au}
\ge
e^{-a\Lbar}\bar G(a)^{-p_0},
\]
because $\displaystyle\int_0^a\mu(s)\,ds = -\log\bar G(a)$. 
Hence
\[
\Psi_u(\Lbar)
\ge
m\int_0^\infty
\mu(a)\bar G(a)^{1-p_0}\,da
=
\frac{m}{1-p_0}
\ge1.
\]
The intermediate value theorem therefore yields a unique root
$\lambda_0(u)\in(0,\Lbar]$.}
\\
(i) \rev{Let $u_1,u_2\ge0$ satisfy
$\|\kappa_{u_i}\|_\infty\le\Lbar$ $(i=1,2)$
and
$u_1\le u_2$. Then
\[
e^{-\int_0^a u_1}
\ge
e^{-\int_0^a u_2},
\]
with strict inequality whenever $u_1<u_2$ on a set of positive
Lebesgue measure. Hence,
\[
\Psi_{u_1}(\lambda)\ge\Psi_{u_2}(\lambda)
\qquad, \forall \lambda\in[0,\Lbar],
\]
with strict inequality under the same condition. Evaluating at
$\lambda=\lambda_0(u_1)$ gives
\[
\Psi_{u_2}\!\bigl(\lambda_0(u_1)\bigr)
\le
\Psi_{u_1}\!\bigl(\lambda_0(u_1)\bigr)
= 1 \equiv \Psi_{u_2}\!\bigl(\lambda_0(u_2)\bigr).
\]
Since $\Psi_{u_2}$ is strictly increasing,
\[
\lambda_0(u_1)\le\lambda_0(u_2),
\]
with strict inequality whenever $u_1<u_2$ on a set of positive
Lebesgue measure.}\\
(ii) \rev{Every $u\in\calU$ satisfies
\[
u(a)<u^\star(a):=\Lbar-p_0\mu(a),
\]
where $u^\star(a)>0$ by~{\rm(H5)}. Since $\|\kappa_{u^\star}\|_\infty=\Lbar$, we have $u^\star\notin\calU$. Moreover, part~(i) yields $\lambda_0(u)\le\lambda_0(u^\star)$,
and therefore
\[
\sup_{u\in\calU}\lambda_0(u)
\le
\lambda_0(u^\star).
\]
To show that this upper bound is sharp, let
\[
u^\star_\varepsilon:=u^\star-\varepsilon,
\qquad
\edm{0<\varepsilon<\Lbar-p_0\mu^\ast=\inf_{a\ge0}u^\star(a).}
\]
\ed{Then $u^\star_\varepsilon\ge0$.}
Since $\|\kappa_{u^\star_\varepsilon}\|_\infty =\Lbar-\varepsilon<\Lbar$, we have $u^\star_\varepsilon\in\calU$. Furthermore,
\[
e^{-\int_0^au^\star_\varepsilon}
=
e^{-\int_0^au^\star}e^{\varepsilon a},
\]
and hence
\[
\Psi_{u^\star_\varepsilon}(\lambda)
=
\Psi_{u^\star}(\lambda+\varepsilon).
\]
Therefore,
\[
\Psi_{u^\star}
\!\left(
\lambda_0(u^\star_\varepsilon)+\varepsilon
\right)
=
1.
\]
Since $\Psi_{u^\star}$ is strictly increasing and admits the unique root
$\lambda_0(u^\star)$, it follows that
\[
\lambda_0(u^\star_\varepsilon)
=
\lambda_0(u^\star)-\varepsilon.
\]
Letting $\varepsilon\downarrow0$ gives
\[
\sup_{u\in\calU}\lambda_0(u)
=
\lambda_0(u^\star),
\]
and the supremum is not attained since $u^\star\notin\calU$.} \\
(iii) \rev{By definition, $\lambda_0(u^\star)$ is the unique solution of
\[
m\int_0^\infty
e^{\lambda a}\mu(a)\bar G(a)
e^{-\int_0^au^\star}\,da
=1.
\]
Using the expression of $u^\star$ and the computations already established
in the proof of the preamble, this equation reduces to
\[
m\int_0^\infty
e^{(\lambda_0(u^\star)-\Lbar)a}
\mu(a)\bar G(a)^{1-p_0}\,da
=1.
\]
Define $\nu^\star:=\lambda_0(u^\star)-\Lbar$, then $\nu^\star$ is the unique solution of~\eqref{eq:nustar}, and hence
$\lambda_0(u^\star)=\Lbar+\nu^\star$.

To determine the sign of $\nu^\star$, we evaluate $\Psi_{u^\star}$ at the
barrier value $\lambda=\Lbar$:
\[
\Psi_{u^\star}(\Lbar)
=
\frac{m}{1-p_0}
=
\E[\xi\mid\xi\ge1]
\ge1,
\]
where the first equality follows from the normalization of the probability density $(1-p_0)\mu(a)\bar G(a)^{1-p_0}$. Since
$\Psi_{u^\star}$ is strictly increasing, then $\lambda_0(u^\star)\le\Lbar$,
that is, $\nu^\star\le0$. Finally,
\[
\nu^\star=0
\iff
\Psi_{u^\star}(\Lbar)=1
\iff
\frac{m}{1-p_0}=1.
\]
Since
\[
m-(1-p_0)
=
\edm{\sum_{n\ge2}(n-1)p_n,}
\]
the latter equality holds if and only if $p_k=0$ for all $k\ge2$, namely $\mathbb P(\xi\ge2)=0$.}
\end{proof}

\begin{remark}[\rev{What is, and what is not, implied by invariance}]
\label{rem:nontrivial}
\rev{Lemma~\ref{lem:invariance} immediately yields the qualitative barrier
$\lambda_0(u)<\Lbar$, so this part is indeed built into the modelling
hypothesis. The contribution of Theorem~\ref{thm:sup} is different: it
identifies the exact supremum $\sup_{u\in\calU}\lambda_0(u)=\Lbar+\nu^\star$,
where the deficit $|\nu^\star|$ is explicitly determined by the offspring law. Three ingredients are needed to obtain this result, and none follows from the invariance property alone. First, the optimisation over the infinite-dimensional
class $\calU$ is reduced, through the monotonicity established in Theorem~\ref{thm:sup}(i), to the single extremal envelope
$u^\star(a)=\Lbar-p_0\mu(a)$. Second, substituting
$\int_0^au^\star=a\Lbar+p_0\log\bar G(a)$ transforms the controlled
Euler--Lotka equation into that of a different Bellman--Harris process, with
hazard $(1-p_0)\mu$ and offspring law $\xi\mid\xi\ge1$. This transformation is
what turns the qualitative barrier $\Lbar$ into the computable quantity
$\Lbar+\nu^\star$. Finally, the sign of $\nu^\star$ follows from
$\E[\xi\mid\xi\ge1]\ge1$, with equality if and only if
$\mathbb P(\xi\ge2)=0$, showing that the gap is entirely determined by the
offspring distribution. For the calibrated model of Section~\ref{sec:num}, the invariance argument
gives $\lambda_0<0.1400$, whereas Theorem~\ref{thm:sup} sharpens this to
$\lambda_0\le0.1050$, reducing the admissible range by approximately $25\%$.}
\end{remark}
\rev{The next corollary compares the two families that are used in practice.}
\begin{corollary}[\rev{Constant versus proportional actuation}]
\label{cor:shape}
\ed{Assume {\rm(H1)--(H6)}.}
\rev{Write $\lambda_0^{\mathrm p}(c):=\lambda_0(c\mu)$ and
$\lambda_0^{\mathrm{cst}}(c_0):=\lambda_0(c_0\mathbf 1)$ for $c,c_0\ge0$. Then
$c\mu\in\calU$ if and only if $(p_0+c)\mu^\ast<\Lbar$, and
$c_0\mathbf 1\in\calU$ if and only if $c_0<\Lbar-\|\kappa\|_\infty$. On the
constant family the output is an exact shift of the uncontrolled decay rate,}
\begin{equation}
\lambda_0^{\mathrm{cst}}(c_0)=\lambda_0(0)+c_0 ,
\label{eq:shift}
\end{equation}
\rev{and the two families are ordered:}
\begin{equation}
\sup_{c\mu\in\calU}\lambda_0^{\mathrm p}(c)
\ \le\
\sup_{c_0\mathbf 1\in\calU}\lambda_0^{\mathrm{cst}}(c_0)
\ =\ \Lbar-\bigl(\|\kappa\|_\infty-\lambda_0(0)\bigr),
\label{eq:order}
\end{equation}
\rev{with strict inequality whenever $\mu<\mu^\ast$ on a set of positive
measure.}
\end{corollary}
\begin{proof}
\rev{The proof is split into three steps.

\pl{\emph{Admissibility criteria.}} $\|\kappa_{c\mu}\|_\infty=\sup_a(p_0+c)\mu(a)
=(p_0+c)\mu^\ast$ and $\|\kappa_{c_0\mathbf 1}\|_\infty=\|\kappa\|_\infty+c_0$,
so both criteria are instances of $\|\kappa_u\|_\infty<\Lbar$.}

\rev{\pl{\emph{Proof of~\eqref{eq:shift}.}} Let $c_0\ge0$. Since $\int_0^ac_0\,ds=c_0a$, the
definition~\eqref{eq:ELu} gives, for every $\lambda$,
\begin{align*}
\Psi_{c_0\mathbf 1}(\lambda)
&=m\int_0^\infty e^{\lambda a}\mu(a)\bar G(a)\,e^{-c_0a}\,da \\
&=m\int_0^\infty e^{(\lambda-c_0)a}\mu(a)\bar G(a)\,da
=\Psi_0(\lambda-c_0).
\end{align*}
}
\rev{Hence $\Psi_{c_0\mathbf 1}(\lambda)=1$ if and only if
$\Psi_0(\lambda-c_0)=1$, that is, if and only if $\lambda-c_0=\lambda_0(0)$ by
uniqueness of the root of $\Psi_0$. This proves~\eqref{eq:shift}.} \\
\rev{Consequently $c_0\mapsto\lambda_0^{\mathrm{cst}}(c_0)$ is increasing and
its supremum over the admissible range $c_0<\Lbar-\|\kappa\|_\infty$ is}
\[
\sup_{c_0\mathbf 1\in\calU}\lambda_0^{\mathrm{cst}}(c_0)
=\lambda_0(0)+\Lbar-\|\kappa\|_\infty
=\Lbar-\bigl(\|\kappa\|_\infty-\lambda_0(0)\bigr),
\]
\rev{which is the value stated in~\eqref{eq:order}. It lies strictly below
$\Lbar$ whenever $\mathbb P(\xi\ge2)>0$, since
$\lambda_0(0)<\|\kappa\|_\infty$ by~\eqref{eq:lam0kappa}. The supremum is not
attained, the admissibility constraint being strict.}

\rev{\pl{\emph{Proof of~\eqref{eq:order}.}} If $c\mu\in\calU$ then, for every $a$,
$c\mu(a)\le c\mu^\ast<\Lbar-p_0\mu^\ast\le\Lbar-p_0\mu(a)$, so
$c\mu\le\Lbar-p_0\mu^\ast$ pointwise, i.e.\ $c\mu\le c_0\mathbf 1$ for
$c_0:=\Lbar-p_0\mu^\ast=\sup\{c_0:c_0\mathbf 1\in\calU\}$. Theorem~\ref{thm:sup}(i)
gives $\lambda_0^{\mathrm p}(c)\le\lambda_0^{\mathrm{cst}}(c_0')$ for every
$c_0'<c_0$ close enough to $c_0$, \pl{which gives}~\eqref{eq:order}. If $\mu<\mu^\ast$ on
a set of positive measure, then $c\mu<c_0\mathbf 1$ there, and the strict part
of Theorem~\ref{thm:sup}(i) applies.}
\end{proof}

\rev{Both families \pl{use up} the same certifiable authority, in the sense that
$\|\kappa_u\|_\infty\uparrow\Lbar$ along each of them, yet neither approaches
the supremum $\Lbar+\nu^\star$ of Theorem~\ref{thm:sup}. The constant family
falls short of it by $\|\kappa\|_\infty-\lambda_0(0)+\nu^\star>0$, and the
proportional family by more; only the envelope-saturating shape $u^\star$
approaches it. What separates the three is the shape of the actuator, not the
certifiable authority it consumes.}



\section{\rev{The proportional actuator: certifiable range, sensitivity, and cost}}
\label{sec:prop}
Theorem~\ref{thm:envelope} bounds every admissible input. \rev{For the
hazard-proportional actuator $u=c\mu$ we now show that this bound is also
sufficient, the controlled process remaining Bellman--Harris so that the
sufficiency result of~\cite{companion} applies unchanged. This yields an
explicit certifiable gain interval, its sensitivity, and the degradation of the
certificate near the admissibility boundary.} \rev{Throughout this section and
Section~\ref{sec:num} we abbreviate
$\lambda_0(c):=\lambda_0^{\mathrm p}(c)=\lambda_0(c\mu)$.}

\subsection{\rev{The certifiable gain interval}}
\label{sec:gain}
Here $u=c\mu$ gives total hazard $(1+c)\mu$, survival
$\bar G_c=\bar G^{\,1+c}$, and an offspring law keeping the natural offspring
with probability $1/(1+c)$:
\begin{equation}
\begin{aligned}
&m_c=\tfrac{m}{1+c},\quad p_{0,c}=\tfrac{p_0+c}{1+c},\\
&\mu^\ast_c=(1+c)\mu^\ast,\quad \mu_{\infty,c}=(1+c)\mu_\infty .
\end{aligned}
\label{eq:eff}
\end{equation}
Hence
$\|\kappa_c\|_\infty=(p_0+c)\mu^\ast$ increases with $c$, whereas
$(1-p_{0,c})\mu_{\infty,c}=\Lbar$ remains constant, reflecting the cancellation
of $(1+c)$ in Lemma~\ref{lem:invariance}.
\rev{The admissibility criterion of Corollary~\ref{cor:shape} for this family,
$(p_0+c)\mu^\ast<\Lbar$, therefore reads $c<c_{\max}$ with}
\begin{equation}
c_{\max}:=\frac{\Lbar-p_0\mu^\ast}{\mu^\ast},
\label{eq:cmax}
\end{equation}
\rev{which equals $1-2p_0$ when $\mu^\ast=\mu_\infty$. The next theorem shows
that this necessary condition is also sufficient.}
\begin{theorem}[Certifiable gain interval]
\label{thm:cmax}
\ed{Assume the uncontrolled process satisfies (H1)--(H6). Then}
under~\eqref{eq:eff}, the controlled process satisfies (H1)--(H5) and (H6)
if and only if \rev{$c\in[0,c_{\max})$, with $c_{\max}$ as
in~\eqref{eq:cmax}}. Hence
\begin{itemize}
    \item[(a)] for every $c<c_{\max}$, all conclusions of
Fact~\ref{fact:selection} hold: the controlled process admits a Yaglom limit
$\nu_Y(c)$, globally attractive and estimable at rate $N^{-\varpi}$; and
\item[(b)] for $c\ge c_{\max}$, \ed{no admissible certificate exists, so no
quasi-stationary equilibrium can be certified}.
\end{itemize}
\end{theorem}
\begin{proof}
By~\eqref{eq:eff}, (H5) becomes
$\|\kappa_c\|_\infty=(p_0+c)\mu^\ast<\Lbar$, i.e.,
$c<c_{\max}$ by~\eqref{eq:cmax}. The remaining assumptions hold on
$[0,c_{\max})$:
(H1) $\bar G^{1+c}$ is absolutely continuous with full support and density
$(1+c)\mu\bar G^{1+c}$ bounded below on compact subsets of $(0,\infty)$;
(H2) $m_c\le m<1$;
(H3) $p_{1,c}=p_1/(1+c)>0$;
(H4) is immediate.
\rev{For (H6), let $\alpha$ be the exponent for $c=0$. The moment condition
is inherited since $p_{n,c}\le p_n$ for every $n\ge1$, so that
$\sum_nn^\alpha p_{n,c}\le\sum_nn^\alpha p_n<\infty$. The controlled barrier is
$\Lbar_c=(1-p_{0,c})\mu_{\infty,c}=\Lbar$ by~\eqref{eq:eff}. The controlled
Malthusian parameter solves $m_c\int e^{\lambda a}\,dG_c=1$; since
$dG_c=(1+c)\mu\bar G^{\,1+c}\,da$ and $m_c=m/(1+c)$, the factors $(1+c)$ cancel
and this is exactly $\Psi_{c\mu}(\lambda)=1$, whose unique root is
$\lambda_0(c)$. \pl{The controlled Malthusian parameter is therefore $\lambda_0(c)$ itself.} Monotonicity
(Theorem~\ref{thm:sup}(i)) applied to $0\le c\mu$ then gives
$\alpha\lambda_0(c)\ge\alpha\lambda_0(0)>\Lbar=\Lbar_c$, the strict inequality
being (H6) for $c=0$; the same $\alpha$ therefore serves for all gains.}
\rev{Part~(a) follows from Fact~\ref{fact:selection}, the controlled process
being Bellman--Harris with parameters~\eqref{eq:eff}. For~(b), let $c\ge c_{\max}$.
Then $\|\kappa_c\|_\infty=(p_0+c)\mu^\ast\ge\Lbar$ by~\eqref{eq:cmax}, whereas
any admissible certificate would give $\|\kappa_c\|_\infty<\lambda_1\le\Lbar$ by
Theorem~\ref{thm:envelope}. As a conclusion, no such certificate exists.}
\end{proof}

\begin{remark}[Certificate versus existence]
\label{rem:cert}
Theorem~\ref{thm:cmax}(b) concerns loss of the certificate, not of the
equilibrium. \rev{At the first-moment level the Malthusian parameter
$\lambda_0(c)$ is still defined by~\eqref{eq:ELu} well beyond $c_{\max}$: by
renewal theory, or by~\cite[Thm.~3.8]{bansaye2020}, and quantitatively up to
$\bar c_{\mathrm{lb}}$ of~\eqref{eq:sep} by Remark~\ref{prop:separation}. What
is lost beyond $c_{\max}$ are the contraction and estimation guarantees of
Fact~\ref{fact:selection}, which rest on the measure-valued lift. Whether the
Yaglom limit $\nu_Y(c)$ itself persists there remains open.}
\end{remark}

\subsection{\rev{Sensitivity and reachable decay rates}}
\label{sec:sens}
\begin{proposition}[Sensitivity and reachable set]
\label{prop:sens}
\ed{Assume {\rm(H1), (H2), (H4), (H5)}.} On $[0,c_{\max})$, the map
$c\mapsto\lambda_0(c)$ is $C^1$, with
\begin{equation}
\lambda_0'(c)=
\frac{\displaystyle\int_0^\infty e^{\lambda_0(c)t}g\,\bar G^{\,c}\,|\log\bar G|\,dt}
{\displaystyle\int_0^\infty t\,e^{\lambda_0(c)t}g\,\bar G^{\,c}\,dt}
\in(0,\mu^\ast],
\label{eq:dlam}
\end{equation}
so $\lambda_0$ is strictly increasing. \rev{Hence the reachable set within the
proportional family is $[\lambda_0(0),\lambda_0^{\max})$, where
$\lambda_0^{\max}:=\lim_{c\uparrow c_{\max}}\lambda_0(c)$ exists by monotonicity
and the bound $\lambda_0\le\Lbar$, and is not attained. Moreover}
\begin{equation}\label{lambda_0}
\lambda_0(c)\le\lambda_0(0)+\mu^\ast c.
\end{equation}
\end{proposition}
\begin{proof}
\rev{
Let
\[
\Psi(c,\lambda)
:=
m\int_0^\infty
e^{\lambda t}g(t)\bar G(t)^c\,dt .
\]
Fix $0\le c_1<c_2<c_{\max}$. Since $\lambda_0(c)\in(0,\Lbar]$ for every
$c$, it suffices to work on $[c_1,c_2]\times[0,\Lbar]$, a fixed compact set.
There the integrand and its partial derivatives in $c$ and $\lambda$ are
bounded in modulus by $\edm{(1+t+\mu^\ast t)}\,e^{\Lbar t}g(t)\bar G(t)^{\,c_1}$,
using $\bar G^{\,c}\le\bar G^{\,c_1}$ and
$|\log\bar G(t)|=\int_0^t\mu\le\mu^\ast t$; this majorant is integrable because $\edm{(1+t+\mu^\ast t)}\,e^{\Lbar t}\le C_\theta e^{\theta t}$
for any $\theta\in(\Lbar,\mu_\infty)$, and $\int e^{\theta t}\,dG<\infty$ by
Fact~\ref{fact:malthus}(i). Hence differentiation under the integral sign is justified, so that $\Psi$ is $C^1$.
\[
\partial_\lambda\Psi(c,\lambda)
=
m\int_0^\infty
t\,e^{\lambda t}g(t)\bar G(t)^c\,dt
>0,
\]
and
\[
\partial_c\Psi(c,\lambda)
=
-m\int_0^\infty
e^{\lambda t}g(t)\bar G(t)^c
|\log\bar G(t)|\,dt
<0.
\]
Since $\lambda_0(c)$ is characterized by
$\Psi(c,\lambda_0(c))=1$, the implicit function theorem yields that
$c\mapsto\lambda_0(c)$ is of class $C^1$, with
\[
\lambda_0'(c)
=
-
\frac{\partial_c\Psi(c,\lambda_0(c))}
{\partial_\lambda\Psi(c,\lambda_0(c))},
\]
which is exactly~\eqref{eq:dlam}. Because
$\partial_\lambda\Psi>0$ and $\partial_c\Psi<0$, one has
$\lambda_0'(c)>0$, so $\lambda_0$ is strictly increasing.

The existence of
\[
\lambda_0^{\max}
=
\lim_{c\uparrow c_{\max}}\lambda_0(c)
\]
follows from monotonicity together with the uniform bound
$\lambda_0(c)\le\Lbar$ (Theorem~\ref{thm:sup}). Since $\edm{c_{\max}\mu\notin\calU}$,
this limit is not attained.

Finally,
\[
|\log\bar G(t)|
=
\int_0^t\mu(s)\,ds
\le
\mu^\ast t,
\]
\pl{so that}
\[
\lambda_0'(c)
\le
\mu^\ast.
\]
Integrating this differential inequality from $0$ to $c$ gives
\[
\lambda_0(c)-\lambda_0(0)
=
\int_0^c\lambda_0'(s)\,ds
\le
\mu^\ast c,
\]
which proves~\eqref{lambda_0}.
}
\end{proof}

\begin{remark}[\rev{Strict separation}]
\label{prop:separation}
\rev{Theorem~\ref{thm:sup}(iii) already separates the reachable set from
$\Lbar$; within this family the separation can also be read off the gain axis.
Only the intermediate-value step of Fact~\ref{fact:malthus}(ii) uses (H5), and
for the controlled process it requires
\[
\Lbar>(1-m_c)\mu_c^\ast=(1+c-m)\mu^\ast
\]
by~\eqref{eq:eff}, that is,
$c<\bar c_{\mathrm{lb}}$ with}
\begin{equation}
\rev{\bar c_{\mathrm{lb}}:=\frac{\Lbar}{\mu^\ast}-(1-m)
=c_{\max}+\edm{\sum_{n\ge2}}(n-1)p_n\ \ge\ c_{\max},}
\label{eq:sep}
\end{equation}
\rev{since 
\[
m-(1-p_0)=\edm{\sum_{n\ge2}}(n-1)p_n,
\]
with equality if and only if $\xi\in\{0,1\}$ almost surely (the other requirement,
$\Lbar<\mu_{\infty,c}$, holds because $1-p_0<1\le1+c$). Thus the ceiling
$\lambda_0(c)<\Lbar$ survives on
$[0,\bar c_{\mathrm{lb}})\supseteq[0,c_{\max})$: the certificate is lost
strictly before any gain in this family could drive $\lambda_0$ to the
transmission barrier.}
\end{remark}

\subsection{\pl{The cost of certifiability}}
\label{sec:cost}
Reference~\cite{companion} removes the particle-number threshold
of~\cite[Thm.~2.3]{cv1}, so~\eqref{eq:rate} holds for every $N\ge2$. What
degrades, unboundedly, is the guaranteed constant.

\rev{We first introduce and fix some constants. \ed{For a gain $c$, let $(V,\lambda_1(c),C_c)$ be an admissible certificate for $X_{c\mu}$ in the sense of Definition~\ref{def:cert};} the associated uniform moment bound of Fact~\ref{fact:selection} is $B_c:=C_c/(\lambda_1(c)-\|\kappa_c\|_\infty)$. We write $\gamma(c)$ for the contraction rate and $\varpi_c:=\gamma(c)/[2(\|\kappa_c\|_\infty+\gamma(c))]\in(0,\tfrac12)$ for the exponent of~\eqref{eq:rate} at gain $c$, and $d_0,C_0$ for the two constants of~\cite[\S4.3]{cv1}, which depend only on the particle system and not on $c$.}

\begin{proposition}[Divergence of the certified guarantee]
\label{prop:cost}
Let $c\in[0,c_{\max})$ and let $\lambda_1(c)$ be the drift rate of any
admissible Lyapunov function for the controlled process. Then
\begin{equation}
\lambda_1(c)-\|\kappa_c\|_\infty\ \le\ \mu^\ast\,(c_{\max}-c),
\label{eq:margin}
\end{equation}
so no admissible certificate can guarantee, through Fact~\ref{fact:selection},
a moment bound smaller than $C_c/[\mu^\ast(c_{\max}-c)]$. \rev{Assume in
addition that the certificate constant and the contraction rate do not
degenerate on the gain range: $\inf_{[0,c_{\max})}C_c>0$ and
$0<\gamma_-\le\gamma(c)\le\gamma_+<\infty$. Then the particle number needed for
a fixed accuracy $\epsilon$ obeys}
\begin{equation}
\rev{N_\epsilon(c)\ \ge\ K_\epsilon(c)\,
(c_{\max}-c)^{-2\|\kappa_c\|_\infty/\gamma(c)}
\ \xrightarrow[c\uparrow c_{\max}]{}\ \infty ,}
\label{eq:cost}
\end{equation}
\rev{where}
\begin{equation}
\rev{K_\epsilon(c):=\Bigl[\frac{\|\kappa_c\|_\infty+\gamma(c)}
{\epsilon\,\|\kappa_c\|_\infty}\,d_0^{\,2\varpi_c}
\Bigl(\frac{C_0C_c}{\mu^\ast}\Bigr)^{1-2\varpi_c}\Bigr]^{1/\varpi_c}.}
\label{eq:Keps}
\end{equation}
\rev{The prefactor $K_\epsilon(c)$ is bounded away from zero uniformly in $c$,
each factor in~\eqref{eq:Keps} being so, and the exponent
in~\eqref{eq:cost} lies in the fixed interval
$[2\|\kappa\|_\infty/\gamma_+,\,2\Lbar/\gamma_-]$; the divergence is therefore
carried by the exponent alone.}
\end{proposition}
\begin{proof}
By Theorem~\ref{thm:envelope} applied to $u=c\mu$, any admissible $\lambda_1(c)$
satisfies $\|\kappa_c\|_\infty<\lambda_1(c)\le\Lbar$; with
$\|\kappa_c\|_\infty=(p_0+c)\mu^\ast$ from~\eqref{eq:eff} and $c_{\max}$
of~\eqref{eq:cmax},
$\Lbar-\|\kappa_c\|_\infty=\Lbar-(p_0+c)\mu^\ast=\mu^\ast(c_{\max}-c)$, which
is~\eqref{eq:margin}; the moment bound $B_c$ is therefore at least
$C_c/[\mu^\ast(c_{\max}-c)]$.
\rev{Minimising $Ae^{\|\kappa_c\|_\infty t}+C_0B_ce^{-\gamma(c)t}$ over $t>0$,
with $A:=d_0N^{-1/2}$, as in~\cite[\S4.3]{cv1}, gives the value
$\frac{\|\kappa_c\|_\infty+\gamma(c)}{\|\kappa_c\|_\infty}
A^{2\varpi_c}(C_0B_c)^{1-2\varpi_c}$, so the constant of~\eqref{eq:rate}
satisfies $d\propto B_c^{\,1-2\varpi_c}$. Then $dN^{-\varpi_c}\le\epsilon$
gives $N\ge(d/\epsilon)^{1/\varpi_c}$, and substituting the lower bound on
$B_c$ yields~\eqref{eq:cost}--\eqref{eq:Keps}, the exponent being
$(1-2\varpi_c)/\varpi_c=2\|\kappa_c\|_\infty/\gamma(c)$. Uniformity follows from
$\|\kappa\|_\infty\le\|\kappa_c\|_\infty<\Lbar$ and
$\gamma_-\le\gamma(c)\le\gamma_+$.}
\end{proof}

Thus $c_{\max}$ is not crossed at constant cost: \rev{by~\eqref{eq:cost} the
certified accuracy degrades at an algebraic rate set by the ratio
$\|\kappa_c\|_\infty/\gamma(c)$}, \pl{so the last part of the admissible range
is of little practical use, even though it is formally certifiable: the gain
limit can be read as a loss of robustness margin.}



\section{Certainty-equivalence feedback}
\label{sec:feedback}
\rev{Sections~\ref{sec:main} and~\ref{sec:prop} answer question~(i) of
Section~\ref{sec:problem}: they say which decay rates can be reached, and at
what certified cost. We now turn to question~(ii) and build a feedback law that
places the conditioned equilibrium at a chosen target. The problem differs from
ordinary regulation in two ways. The plant does not need to be stabilised, so
the loop \pl{picks} the attractor instead of creating it. And the regulated output
is never measured directly: all we have is the particle estimate
$\widehat\lambda_0^N=\mathcal X^N(\kappa_{c\mu})$, with the accuracy given by
Fact~\ref{fact:selection}.}

\rev{This suggests a certainty-equivalence law: use the estimate in place of
the unknown $\lambda_0(c)$, and integrate the output error. One more ingredient
is needed. The gain must stay in a compact interval $[0,c_1]$ with
$c_1<c_{\max}$, and the reason is not actuator saturation. By
Theorem~\ref{thm:cmax}(b) there is no admissible certificate once
$c\ge c_{\max}$, so the estimate loses its guarantee exactly where the input
constraint becomes active. We therefore enforce the constraint of
Theorem~\ref{thm:envelope} inside the loop, as a projection. The margin
$c_{\max}-c_1$ is what pays for the sensing accuracy, as
Proposition~\ref{prop:cost} shows.}

\subsection{\rev{The projected gain law}}
\label{sec:gainlaw}
\begin{proposition}[\rev{ISS of the certainty-equivalence gain law}]
\label{prop:feedback}
Fix $c_1<c_{\max}$ and let $\ell:=\min_{c\in[0,c_1]}\lambda_0'(c)>0$, which exists and
is positive by Proposition~\ref{prop:sens}. Let
$\lambda^\star\in[\lambda_0(0),\lambda_0(c_1)]$ and
$c^\star:=\lambda_0^{-1}(\lambda^\star)$. Suppose the estimator error is
uniformly bounded on the gain range,
\begin{equation}
\rev{\bigl|\widehat\lambda_0^N(c)-\lambda_0(c)\bigr|\le e_N
\qquad\text{for all }c\in[0,c_1],}
\label{eq:esterr}
\end{equation}
and let $\widehat c(\cdot)$ solve the projected gain update
\begin{equation}
\dot{\widehat c}=k\,\Pi_{[0,c_1]}\!\bigl(\widehat c,\;\lambda^\star-\widehat\lambda_0^N(\widehat c)\bigr),
\qquad k>0,
\label{eq:gainlaw}
\end{equation}
\rev{where $\Pi_{[0,c_1]}$ zeroes the velocity component pointing out of
$[0,c_1]$. Then $[0,c_1]$ is forward invariant and, for every
$\widehat c(0)\in[0,c_1]$,}
\begin{equation}
\rev{|\widehat c(t)-c^\star|\ \le\ |\widehat c(0)-c^\star|\,e^{-\frac{k\ell}{2} t}
\ +\ 2\ell^{-1}e_N ,\quad t\ge0 .}
\label{eq:iss}
\end{equation}
\rev{Hence~\eqref{eq:gainlaw} is ISS with respect to the estimation error, with
linear asymptotic gain $2\ell^{-1}$, and any equilibrium of~\eqref{eq:gainlaw}
in the interior of $[0,c_1]$ satisfies the sharper static bound
$|\widehat c-c^\star|\le\ell^{-1}e_N$.}
\end{proposition}

\begin{proof}
\rev{Write $e(t):=\widehat c(t)-c^\star$, $W:=\frac{1}{2}e^2$, and
$v:=\lambda^\star-\widehat\lambda_0^N(\widehat c)$ for the unprojected velocity
direction, so that~\eqref{eq:gainlaw} reads
$\dot{\widehat c}=k\,\Pi_{[0,c_1]}(\widehat c,v)$.}

\rev{\emph{Step 1: the projection never increases $|e|$.} By construction
$\Pi_{[0,c_1]}(\widehat c,v)=v$ unless $\widehat c=0$ with $v<0$, or
$\widehat c=c_1$ with $v>0$, in which cases it vanishes. In the first case
$e=-c^\star\le0$ and $v<0$, so $ev\ge0$; in the second $e=c_1-c^\star\ge0$ and
$v>0$, so again $ev\ge0$. In both cases $e\,\Pi_{[0,c_1]}(\widehat c,v)=0\le ev$.
Therefore}
\begin{equation}
\rev{\dot W=e\,\dot{\widehat c}=k\,e\,\Pi_{[0,c_1]}(\widehat c,v)\ \le\ k\,e\,v ,}
\label{eq:proj}
\end{equation}
\rev{and $[0,c_1]$ is forward invariant, so $\widehat c(t)$ stays in the range
where~\eqref{eq:esterr} holds.}

\rev{\emph{Step 2: decomposition of the velocity.} Since
$\lambda^\star=\lambda_0(c^\star)$,}
\[
v=-\underbrace{\bigl[\lambda_0(\widehat c)-\lambda_0(c^\star)\bigr]}_{=:A}
+\underbrace{\bigl[\lambda_0(\widehat c)-\widehat\lambda_0^N(\widehat c)\bigr]}_{=:\delta} ,
\]
\rev{where $|\delta|\le e_N$ by~\eqref{eq:esterr}. By
Proposition~\ref{prop:sens}, $\lambda_0$ is $C^1$ on $[0,c_1]$ with
$\lambda_0'\ge\ell$, so the mean value theorem gives
$A=\lambda_0'(\xi)\,e$ for some $\xi$ between $c^\star$ and $\widehat c$, \pl{so that}}
\[
eA=\lambda_0'(\xi)\,e^2\ \ge\ \ell\,e^2 .
\]

\rev{\emph{Step 3: a decay inequality outside a residual set.} Combining the
two displays with $e\delta\le|e|\,e_N$,}
\[
ev=-eA+e\delta\ \le\ -\ell\,e^2+|e|\,e_N
=-|e|\bigl(\ell|e|-e_N\bigr).
\]
\rev{Assume $|e|\ge\frac{2e_N}{\ell}$. Then $e_N\le\frac{\ell|e|}{2}$, so
$\ell|e|-e_N\ge\frac{\ell|e|}{2}$ and}
\[
ev\ \le\ -\frac{\ell}{2}\,e^2=-\ell\,W .
\]
\rev{With~\eqref{eq:proj} this yields}
\begin{equation}
\rev{\dot W\ \le\ -k\ell\,W
\qquad\text{whenever }|e|\ge\frac{2e_N}{\ell}.}
\label{eq:decay}
\end{equation}

\rev{\emph{Step 4: invariance of the residual set and conclusion.} Let
\[
S:=\bigl\{|e|\le\frac{2e_N}{\ell}\bigr\} \text{ and } T:=\inf\{t\ge0:\ e(t)\in S\},
\]
with $T=+\infty$ if the infimum is over an
empty set. On $[0,T)$ the bound~\eqref{eq:decay} applies, and the comparison
lemma gives $W(t)\le W(0)e^{-k\ell t}$, that is,}
\[
|e(t)|\ \le\ |e(0)|\,e^{-\frac{k\ell t}{2}},\qquad t<T .
\]
\rev{On the boundary of $S$ we have $\dot W\le-k\ell W<0$ by~\eqref{eq:decay},
so $S$ is forward invariant and $|e(t)|\le\frac{2e_N}{\ell}$ for all $t\ge T$.
In either case}
\[
|e(t)|\ \le\ \max\Bigl\{|e(0)|\,e^{-\frac{k\ell t}{2}},\ \frac{2e_N}{\ell}\Bigr\}
\ \le\ |e(0)|\,e^{-\frac{k\ell t}{2}}+\frac{2e_N}{\ell},
\]
\rev{which is~\eqref{eq:iss}.}

\rev{\emph{Step 5: interior equilibria.} If $\widehat c$ is an equilibrium
in the interior of $[0,c_1]$ then $\Pi_{[0,c_1]}$ acts as the identity, so
$v=0$, i.e.\ $\widehat\lambda_0^N(\widehat c)=\lambda^\star$. Hence
$A=-\delta$ and, by Steps~2,
$\ell|e|\le|A|=|\delta|\le e_N$, that is,
$|e|\le\frac{e_N}{\ell}$.}
\end{proof}

\rev{The regulated variable is not the gain but the decay rate it produces, and
for that variable the sensitivity constant disappears.}
\begin{corollary}[\rev{Placement error}]
\label{cor:placement}
\rev{Under the hypotheses of Proposition~\ref{prop:feedback},}
\begin{equation}
\rev{\bigl|\lambda_0(\widehat c(t))-\lambda^\star\bigr|
\ \le\ \bigl|\lambda_0(\widehat c(0))-\lambda^\star\bigr|\,
e^{-\frac{k\ell t}{2}}+2e_N ,\quad t\ge0 .}
\label{eq:placement}
\end{equation}
\rev{The asymptotic placement error is therefore at most twice the estimator
error, however sensitive $\lambda_0$ may be to the gain.}
\end{corollary}
\begin{proof}
\rev{Write $y(t):=\lambda_0(\widehat c(t))-\lambda^\star$,
$V:=\frac{1}{2}y^2$, and let $v:=\lambda^\star-\widehat\lambda_0^N(\widehat c)$
be as in the proof of Proposition~\ref{prop:feedback}.}

\rev{\emph{Step 1: a scalar equation for the output error.} Since
$\lambda^\star=\lambda_0(c^\star)$, Step~2 of that proof gives $v=-y+\delta$,
where $\delta:=\lambda_0(\widehat c)-\widehat\lambda_0^N(\widehat c)$ and
$|\delta|\le e_N$ by~\eqref{eq:esterr}. Differentiating $y$
along~\eqref{eq:gainlaw},}
\[
\dot y=\lambda_0'(\widehat c)\,\dot{\widehat c}
=k\,\lambda_0'(\widehat c)\,\Pi_{[0,c_1]}(\widehat c,v).
\]

\rev{\emph{Step 2: the projection never increases $|y|$.} The projection is
inactive except at $\widehat c=0$ with $v<0$, or at $\widehat c=c_1$ with
$v>0$. In the first case $y=\lambda_0(0)-\lambda^\star\le0$, since
$\lambda^\star\ge\lambda_0(0)$; in the second
$y=\lambda_0(c_1)-\lambda^\star\ge0$, since $\lambda^\star\le\lambda_0(c_1)$.
Either way $yv\ge0$, while $y\,\Pi_{[0,c_1]}(\widehat c,v)=0$. Hence
$y\,\Pi_{[0,c_1]}(\widehat c,v)\le yv$ at all times, and $\lambda_0'>0$ gives}
\begin{equation}
\rev{\dot V=y\dot y\ \le\ k\,\lambda_0'(\widehat c)\,y\,v .}
\label{eq:projy}
\end{equation}

\rev{\emph{Step 3: decay outside a residual set.} From $v=-y+\delta$ and
$y\delta\le|y|\,e_N$,}
\[
yv=-y^2+y\delta\ \le\ -|y|\bigl(|y|-e_N\bigr).
\]
\rev{Suppose $|y|\ge2e_N$. Then $e_N\le\frac{|y|}{2}$, so
$|y|-e_N\ge\frac{|y|}{2}$ and $yv\le-\frac{y^2}{2}=-V$. Since
$\lambda_0'\ge\ell$ on $[0,c_1]$ and $yv\le0$, inequality~\eqref{eq:projy}
yields}
\begin{equation}
\rev{\dot V\ \le\ -k\ell\,V
\qquad\text{whenever }|y|\ge2e_N .}
\label{eq:decayy}
\end{equation}

\rev{\emph{Step 4: conclusion.} Let $S:=\{|y|\le2e_N\}$ and
$T:=\inf\{t\ge0:\ y(t)\in S\}$, with $T=+\infty$ if $y$ never enters $S$. On
$[0,T)$ the comparison lemma applied to~\eqref{eq:decayy} gives
$V(t)\le V(0)e^{-k\ell t}$, that is,
$|y(t)|\le|y(0)|e^{-\frac{k\ell t}{2}}$. On the boundary of $S$ we have
$\dot V<0$ by~\eqref{eq:decayy}, so $S$ is forward invariant and
$|y(t)|\le2e_N$ for $t\ge T$. In both regimes}
\[
|y(t)|\ \le\ \max\Bigl\{|y(0)|\,e^{-\frac{k\ell t}{2}},\ 2e_N\Bigr\}
\ \le\ |y(0)|\,e^{-\frac{k\ell t}{2}}+2e_N ,
\]
\rev{which is~\eqref{eq:placement}.}
\end{proof}

\subsection{\rev{Tuning, and the role of the estimator}}
\label{sec:tuning}
\rev{Proposition~\ref{prop:feedback} and Corollary~\ref{cor:placement} treat
$\widehat\lambda_0^N$ as an external signal of known accuracy. Two questions are
left open by that reading: how fast the loop may be run before the estimator
can no longer follow it, and how the bound~\eqref{eq:esterr} relates to what
Fact~\ref{fact:selection} actually provides. We take them in turn.}

\subsubsection*{Choice of $k$ and timescale separation}
\rev{The residual terms in~\eqref{eq:iss} and~\eqref{eq:placement} do not
depend on $k$; the gain only fixes the rate $\frac{k\ell}{2}$ at which the
initial error is forgotten. Nothing in the analysis therefore prevents taking
$k$ arbitrarily large, \pl{which cannot be correct}: the particle system needs a time
of order $\frac{1}{\gamma(c)}$ to settle, and~\eqref{eq:esterr} describes its
stationary empirical measure, not its transient. The law~\eqref{eq:gainlaw} is
thus meaningful only when $k\ell\ll\gamma(c)$, so that the estimator can follow
the current gain. Removing this separation would require the
singular-perturbation framework of~\cite{kokotovic} together with the ISS
small-gain theorem of~\cite{jiangteel}, and would also need a lower bound on
$\gamma(c)$, which is not available (Section~\ed{\ref{sec:concl}}).}

\subsubsection*{From the particle bound to~\eqref{eq:esterr}}
\rev{Hypothesis~\eqref{eq:esterr} holds pathwise and uniformly in $c$, whereas
Fact~\ref{fact:selection} gives, for each fixed $c$, a bound in expectation:
$\E|\widehat\lambda_0^N(c)-\lambda_0(c)|\le d\|\kappa_c\|_\infty N^{-\varpi}
\le d\,\Lbar\,N^{-\varpi}$ by~\eqref{eq:env0}. Read $e_N$ as the random
quantity $\sup_{c\in[0,c_1]}|\widehat\lambda_0^N(c)-\lambda_0(c)|$. Taking
expectations in~\eqref{eq:iss} and~\eqref{eq:placement} then gives}
\begin{align}
\rev{\limsup_{t\to\infty}\E\bigl|\lambda_0(\widehat c(t))-\lambda^\star\bigr|}
&\le2\,\E[e_N],\label{eq:expbounds1}\\
\limsup_{t\to\infty}\E\bigl|\widehat c(t)-c^\star\bigr|
&\le\frac{2\,\E[e_N]}{\ell}.\label{eq:expbounds2}
\end{align}
\rev{The gap is that Fact~\ref{fact:selection} controls $\sup_{c}\E|\widehat\lambda_0^N(c)-\lambda_0(c)|$, whereas~\eqref{eq:expbounds1}--\eqref{eq:expbounds2} need $\E\bigl[\sup_{c}|\cdot|\bigr]$, a strictly stronger quantity.}

\rev{One hypothesis of the design is not covered by either remark. Placing the
decay rate is only useful if it places the equilibrium itself, and this
requires $c\mapsto\nu_Y(c)$ to be TV-Lipschitz on $[0,c_1]$: with constant
$L_\nu$, \eqref{eq:expbounds2} gives}
\[
\rev{\limsup_{t\to\infty}
\E\bigl\|\nu_Y(\widehat c(t))-\nu_Y(c^\star)\bigr\|_{\mathrm{TV}}
\le\frac{2L_\nu\,\E[e_N]}{\ell}.}
\]
\rev{No quantitative form of this Lipschitz property is available at present,
and we state it as an assumption rather than derive it.}



\section{Numerical illustration: foot-and-mouth surveillance}
\label{sec:num}
We use the model calibrated in~\cite{companion} to the 2001 north Cumbrian
foot-and-mouth outbreak. An individual is an infected premises, its lifetime the
infectious period, its offspring the premises it infects, so
$m=R_{\mathrm{eff}}$. The infectious period is $\Gamma(2,\theta)$ with
$\E[T]=8$~d, hence $\theta=4$~d and $\mu(a)=a/[\theta(\theta+a)]$ increasing to
$\mu^\ast=\mu_\infty=0.25$~d$^{-1}$; offspring are Poisson, so $p_0=e^{-m}$; the
decline of the surveillance record gives $\lambda_0=0.0235$~d$^{-1}$ and, by
Euler--Lotka inversion $m=(1-\lambda_0\theta)^2$, $m=0.8208$ and $p_0=0.4401$.
Then $\|\kappa\|_\infty=0.1100$ and $\Lbar=0.1400$~d$^{-1}$, so (H5) holds with
margin $0.0300$~d$^{-1}$; (H6) holds for every $\alpha>5.96$ (the companion's
weaker form requires only $\alpha>\edm{4.68}$), the Poisson law having all moments.

\emph{The envelope.} By~\eqref{eq:envelope} the certifiable removal hazard is
bounded by $u^\star(a)=0.1400-0.4401\,\mu(a)$, falling from $0.140$~d$^{-1}$ at
age zero to $0.030$~d$^{-1}$ at large ages (Table~\ref{tab:env},
Fig.~\ref{fig:authority}A): a freshly infected premises may be culled
preventively at up to $0.14$ per day, a long-standing one at only $0.03$, because
at large ages the natural childless removal $p_0\mu(a)$ has already consumed
almost all the margin.
\begin{table}[H]
\centering
\caption{Certifiable removal envelope $u^\star(a)=\Lbar-p_0\mu(a)$, calibrated
FMD model.}
\label{tab:env}
\footnotesize
\setlength{\tabcolsep}{5pt}
\begin{tabular}{@{}|l|c|c|c|c|c|c|@{}}
\hline
age $a$ (d) & $0$ & $2$ & $4$ & $8$ & $16$ & $\infty$ \\ \hline
$\mu(a)$ (d$^{-1}$) & $0$ & $0.083$ & $0.125$ & $0.167$ & $0.200$ & $0.250$\\\hline
$u^\star(a)$ (d$^{-1}$) & $0.140$ & $0.103$ & $0.085$ & $0.067$ & $0.052$ & $0.030$\\
\hline
\end{tabular}
\end{table}

\emph{\pl{Comparing the three shapes.}} Solving~\eqref{eq:ELu} by quadrature for the three
shapes gives Table~\ref{tab:ladder}. The proportional actuator, admissible up to
$c_{\max}=1-2p_0=0.1199$ by~\eqref{eq:cmax}, attains
$\lambda_0=0.0390$~d$^{-1}$, i.e.\ $27.8\%$ of the barrier; the constant
actuator, admissible up to $0.0300$~d$^{-1}$, attains $0.0535$ exactly, in
accordance with the shift identity~\eqref{eq:shift} of
Corollary~\ref{cor:shape}; \rev{and the envelope-saturating value
$\Lbar+\nu^\star=0.1050$ of Theorem~\ref{thm:sup}(iii), i.e.\ $75.0\%$, is
approached but not attained, with $\nu^\star=-0.0351$ and
$\E[\xi\mid\xi\ge1]=1.466$.} \pl{Reshaping the actuator gives a factor $2.69$ in
decay rate at the same level of certifiability, and brings the mean persistence
of the conditioned outbreak from $25.7$ down to $9.5$ days.}
\begin{table}[t]
\centering
\caption{Reachable decay rates by actuator shape, calibrated FMD model. All
entries by quadrature on~\eqref{eq:ELu}; barrier $\Lbar=0.1400$~d$^{-1}$.}
\label{tab:ladder}
\footnotesize
\setlength{\tabcolsep}{4pt}
\begin{tabular}{@{}lccc@{}}
\hline
actuator & $\sup\lambda_0$ (d$^{-1}$) & \% of $\Lbar$ & $1/\lambda_0$ (d)\\
\hline
none, $u\equiv0$                       & $0.0235$ & $16.8$ & $42.6$\\
proportional, $u=c\mu$, $c\to c_{\max}$ & $0.0390$ & $27.8$ & $25.7$\\
constant, $u\equiv c_0\to0.030$         & $0.0535$ & $38.2$ & $18.7$\\
envelope-saturating, $u\to u^\star$     & $0.1050$ & $75.0$ & $9.5$\\
\hline
barrier (unattainable)                  & $0.1400$ & $100$  & $7.1$\\
\hline
\end{tabular}
\end{table}

\emph{The proportional family in detail.} Here $\lambda_0'(0)=0.130$ and
$\lambda_0'\approx0.128$ as $c\uparrow c_{\max}$, well below the \pl{rough} bound
$\mu^\ast=0.25$ of Proposition~\ref{prop:sens}, and
$\lambda_0(0)+\mu^\ast c_{\max}=0.0535$ is exactly the constant-actuator
supremum, as Corollary~\ref{cor:shape} predicts. Numerically $\lambda_0(c)$
meets the barrier only at $\bar c=0.9353$, about $7.8$ times $c_{\max}$ and well
past the separation bound $\bar c_{\mathrm{lb}}=0.3808$ of
\rev{Remark~\ref{prop:separation}}: within this family the actuator is stopped
by certifiability, not by physics. By Proposition~\ref{prop:cost} the certified
margin $\mu^\ast(c_{\max}-c)$ falls from $0.0300$ at $c=0$ to $0.0025$ at
$c=0.11$ and $0.0012$ at $c=0.115$, \rev{\pl{so the guaranteed moment bound
becomes twelve and twenty-five times larger}}.

\rev{\emph{Feedback.} Take $c_1=0.11$, leaving a certified margin
$\mu^\ast(c_{\max}-c_1)=0.0025$~d$^{-1}$ and giving
$\ell=\min_{[0,c_1]}\lambda_0'=0.128$. By Corollary~\ref{cor:placement} the
asymptotic placement error on the decay rate is at most $2e_N$, whereas
by~\eqref{eq:iss} the gain only settles within
$\frac{2e_N}{\ell}=15.6\,e_N$ of $c^\star$. The loop therefore inherits the
estimator accuracy on the regulated output without amplification, and pays the
factor $\ell^{-1}$ only on the gain itself.}

\begin{figure}[t]
\centering
\includegraphics[width=1\linewidth]{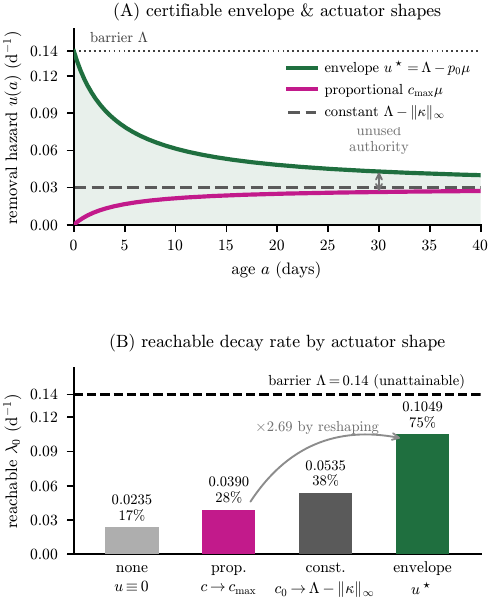}
\caption{(A) The certifiable envelope $u^\star(a)=\Lbar-p_0\mu(a)$ (shaded:
admissible) and the three actuator shapes. The proportional command increases
with age while the envelope decreases, \pl{so it saturates at large ages and
leaves authority unused at small ones}. (B) The resulting decay rates against the barrier
$\Lbar$, invariant under childless actuation; the supremum over all certifiable
actuators is $\Lbar+\nu^\star$ (Thm.~\ref{thm:sup}), \ed{approached, but not
attained,} by saturating the envelope.}
\label{fig:authority}
\end{figure}

\emph{Scope.} Under Poisson offspring with $\mu^\ast=\mu_\infty$, (H5) reads
$m>\log2$, so with (H2) the admissible range is
$R_{\mathrm{eff}}\in(\log2,1)=(0.693,1)$: \pl{the weakly subcritical range,
where an outbreak lasts long enough for $\nu_Y$ to be the right reference}, and
where our calibration sits. Within that band $c_{\max}=1-2e^{-m}$ is
increasing in $m$, from $0$ at $m\downarrow\log2$ to $1-2e^{-1}=0.264$ at
$m\uparrow1$: certifiable authority is \pl{smallest} where the population is most
strongly subcritical. The restriction $m>\log2$ is a property of the
measure-valued lift~\cite[Rem.~3.10]{companion}, not a statement about the
epidemic. \rev{\pl{In practice, removal gives a limited speed-up and an
equilibrium that can be certified and monitored, whereas vaccination or movement
restrictions act on $m$ and move the ceiling itself.}}



\section{Conclusions and Prospects}
\label{sec:concl}
\rev{We have shown that preventive removal acting on a subcritical age-structured
branching population has a \pl{limited} authority over its conditioned
equilibrium\pl{, and that the limit comes from the structure of the problem}. The input is matched to the killing rate but unmatched with respect
to the Foster--Lyapunov drift (Lemma~\ref{lem:invariance}), which yields a
closed-form, control-law-independent input constraint set
(Theorem~\ref{thm:envelope}) and, through the monotonicity of the input--output
map, the exact supremum $\Lbar+\nu^\star$ of the reachable decay rates
(Theorem~\ref{thm:sup}). The deficit $|\nu^\star|$ is explicit and is set by the
offspring law alone. Along the admissibility boundary the achievable decay rate
is \pl{set} by the shape of the actuator rather than its \pl{size}
(Corollary~\ref{cor:shape}), a factor of $2.69$ on the calibrated foot-and-mouth
model, and the certainty-equivalence law of Section~\ref{sec:feedback} places
the equilibrium with an error proportional to that of the particle estimator.

\pl{Three restrictions apply.} (i) \emph{Actuator class.} The
ceiling is a limitation of \pl{childless} actuation, not of feedback \pl{in
general}: an actuator that \pl{reduces} onward transmission acts on $m$, hence on $\Lbar$ itself,
and Theorem~\ref{thm:sup} does not apply to it. (ii) \emph{Necessity versus
sufficiency.} Theorem~\ref{thm:envelope} is a necessary condition, so $\calU$ is
an outer bound on $\calU_{\mathrm c}$ and Theorem~\ref{thm:sup} an upper bound
on certifiable performance; a matching sufficiency for a general $u$ requires
the age-dependent controlled reproductive value of~\eqref{eq:agecancel}, and the
certified statements of Section~\ref{sec:prop} are accordingly \pl{limited} to the
proportional family. (iii) \emph{Certificate versus equilibrium}, as \pl{set out}
in Remark~\ref{rem:cert}.

\pl{The largest extension} concerns the supercritical regime $m>1$, where
(H2) fails, the population survives with positive probability, and the Yaglom
limit is no longer the relevant operating point. \pl{There the control problem
becomes a different one}: the plant is unstable and the actuator must
\pl{stabilise} it rather than place an existing attractor. \pl{Proportional
removal does stabilise it}, since by~\eqref{eq:eff} the controlled offspring mean is $m_c=m/(1+c)$, so the
closed loop is subcritical as soon as $c>m-1$; combined with the certifiability
constraint $c<c_{\max}$ of Theorem~\ref{thm:cmax}, a certifiably stabilising
gain exists if and only if $m-1<c_{\max}$. When $\mu^\ast=\mu_\infty$ this reads
$m+2p_0<2$, which for Poisson offspring holds up to $m^\star\simeq1.594$; beyond
that threshold childless removal cannot certifiably stabilise the population,
whatever the gain. Turning this observation into a theorem requires replacing
the quasi-stationary framework by a \pl{full} stability analysis of the
measure-valued flow, and the transient guarantees of
Section~\ref{sec:feedback} by a stabilisation certificate valid before the
closed loop becomes subcritical.

\pl{Further open questions are} the sufficiency direction for a general $u$; the
optimal-shape problem, of which Corollary~\ref{cor:shape} compares only two
families and Theorem~\ref{thm:sup} solves only the extreme case, \pl{when
$\int u(s){\rm d}s$ is bounded instead of $u$ itself}; an explicit bound on
$\gamma(c)$, for which the non-conservative Harris framework
of~\cite{bansaye2022} \pl{is the natural tool} and which would make
Proposition~\ref{prop:cost} unconditional; and a central limit theorem for the
estimator~\cite{cerou} in place of the present first-moment rate.}



\end{document}